\documentclass[runningheads]{llncs}
\pdfpagesattr{/CropBox [92 67 523 728]}

\usepackage[utf8]{inputenc}
\usepackage{verbatim}
\usepackage{amsmath,amssymb}
\usepackage[sort]{cite}
\usepackage{paralist}
\usepackage{url}
\usepackage{xspace}
\newcommand{\ie}{\textit{i.e.},~}
\usepackage{graphicx}
\newcommand{\coll}[1]{\url{#1}}
\usepackage{xcolor}
\usepackage{textcomp}

\newcommand{\vertex}[2]{\ensuremath{v_{#1}^{#2}}}
\newcommand{\pair}[2]{\ensuremath{[#1, #2]}}

\newcommand{\threePCRP}{\ensuremath{3}-PCRP\xspace}
\newcommand{\twoPCRP}{\ensuremath{2}-PCRP\xspace}
\newcommand{\kPCRP}{$k$-PCRP\xspace}
\newcommand{\OP}{\ensuremath{\mathit{OP}}}

\title{Covering Pairs in Directed Acyclic Graphs}

\author{%
Niko Beerenwinkel\inst{1} \and
Stefano Beretta\inst{2,4} \and
Paola Bonizzoni\inst{2} \and
Riccardo Dondi\inst{3} \and
Yuri Pirola\inst{2}}
\authorrunning{Beerenwinkel \textit{et al.}}

\institute{%
Dept.~of Biosystems Science and Engineering, ETH Zurich, Basel,
Switzerland, \email{niko.beerenwinkel@bsse.ethz.ch}
\and
DISCo, Univ.~degli Studi di Milano-Bicocca, Milan, Italy,
\email{\{beretta,bonizzoni,pirola\}@disco.unimib.it}
\and
Dip.~di Scienze Umane e Sociali, Univ.~degli Studi di Bergamo, Bergamo,
Italy, \email{riccardo.dondi@unibg.it}
\and
Inst.~for Biomedical Technologies, National Research Council, Segrate,
Italy
}

\begin{document}
\maketitle

\begin{abstract}
The Minimum Path Cover problem on directed acyclic graphs (DAGs) is a
classical problem that provides a clear and simple
mathematical formulation for several applications in different areas and that
has an efficient algorithmic solution.
% Since its first studies, several variants have been proposed in order to
% model constraints often arising in practice.
In this paper, we study the computational complexity of two constrained
variants of Minimum Path Cover motivated by the recent introduction of
next-generation sequencing technologies in bioinformatics.
% The first problem (MinPCRP), given a DAG and a set of
% pairs of vertices, asks for a minimum cardinality set of paths such that
% each vertex belongs to one of the paths and such that both vertices
% of each pair belong to the same path.
The first problem (MinPCRP), given a DAG and a set of pairs of
vertices, asks for a minimum cardinality set of paths ``covering'' all the
vertices such that both vertices of each pair belong to the same
path.
For this problem,
% we establish a sharp tractability borderline by showing that,
we show that,
while it is NP-hard to compute if there exists a solution
consisting of at most three paths, it is possible to decide in
polynomial time whether
% there exists a solution consisting of at most two paths.
a solution consisting of at most two paths exists.
The second problem (MaxRPSP), given a DAG and a set of pairs of
vertices, asks for a path containing the maximum number of the
given pairs of vertices.
% An efficient solution of this problem would clearly provide an efficient
% heuristic solution for the first problem.
% However, we show that this problem is NP-hard and, moreover, also W[1]-hard
% when the parameter is the size of the solution, \ie the number of
% ``covered'' pairs.
We show its NP-hardness and also its W[1]-hardness
when parametrized by the number of covered pairs.
On the positive side, we give a fixed-parameter algorithm when the
parameter is the maximum overlapping degree, a natural parameter in the
bioinformatics applications of the problem.

% In this paper we study two combinatorial problems related to a
% constrained variant of Minimum Path Cover,
% motivated mainly by applications in the analysis of Next-Generation
% Sequencing data in bioinformatics.
% The first problem (MinPCRP), given a directed acyclic graph and a
% set of pairs of vertices,
% asks for a minimum cardinality set of paths such that
% each vertex belongs to one of the paths and the vertices of each pair
% belong to a same path.
% In this paper we investigate the complexity of the problem,
% and we show that, while it is NP-hard to compute if there exists a
% solution consisting of at most three paths,
% it is polynomial time solvable to compute the existence of a
% solution consisting of at most two paths.
% Then, aiming at designing a heuristic approach for the MinPCRP problem,
% we introduce and investigate a combinatorial problem that,
% given a directed acyclic graph and a set of pairs of vertices,
% asks for a path that contains the maximum number of pairs of vertices.
% We investigate the complexity of the problem and we show that it is not
% only NP-hard, but also W[1]-hard when the parameter is the size of the
% solution, \ie the number of covered pairs.
% On the positive side, we give a fixed-parameter algorithm when the parameter
% is the number of overlapping pairs,
% a natural parameter in some applications of the problem.
\end{abstract}

\section{Introduction}
\label{sec:intro}
The \emph{Minimum Path Cover} (MinPC) problem is a well-known problem in
graph theory~\cite{fordfulkerson}.
Given a \emph{directed acyclic graph} (DAG), MinPC
asks for a minimum-cardinality set $\Pi$ of paths
%%of a given \emph{directed} graph $G$
such that each vertex of $G$ belongs to at least one path of $\Pi$.
The problem can be solved in polynomial time %on \emph{directed
%acyclic graphs} (DAGs)
with an algorithm based on the well-known Dilworth's theorem for partially ordered
sets~\cite{Dilworth1950}, which
allows to relate the size of a minimum path cover to that of a
maximum matching in a bipartite graph obtained from the input DAG.
%%This algorithm can also solve the problem on generic directed graphs
%%since strongly connected components, which can be covered by one path,
%%can be safely replaced by a single vertex.

The Minimum Path Cover problem has important applications
in several fields ranging from
bioinformatics~\cite{Beerenwinkel2008,Trapnell2010,bao2013branch} to
software testing~\cite{Ntafos1979}.
In particular, in bioinformatics the Minimum Path Cover problem is
applied to the reconstruction of a set of highly-similar sequences
starting from a large set of their short fragments (called \emph{short
reads})~\cite{Trapnell2010,Beerenwinkel2008}.
More precisely, each fragment is represented by a single vertex and two
vertices are connected if the alignments of the corresponding reads on
the genomic sequence overlap.
In~\cite{Trapnell2010}, the paths on such a graph represent putative
transcripts and a minimum-cardinality set of paths ``covering'' all the
vertices represents a set of protein isoforms which are likely to
originate from the observed reads.
On the other hand, in~\cite{Beerenwinkel2008} the paths on such a graph
represent the genomes of putative viral haplotypes and a
minimum-cardinality set of paths covering the whole graph represents the
likely structure of a viral population.

Recently, different constraints have motivated the definition of new
variants of the minimum path cover problem.
In~\cite{bao2013branch}, given a DAG $D$ and a set $P$ of required
paths, the proposed problem asks for a minimum cardinality set of paths
such that: (1)~each vertex of the graph belongs to some path, and
(2)~each path in $P$ is a subpath of a path of the solution.
The authors have described a polynomial-time algorithm to solve this
problem by collapsing each required path into a single vertex and then
finding a minimum path cover on the resulting graph.
Other constrained problems related to minimum path cover have been
proposed in the context of social network analysis and, given an edge-colored graph, ask for the
maximum number of vertex-disjoint uni-color paths that cover the
vertices of the given graph~\cite{Wu12MaxCDP,Bonizzoni13Alg}.
%% Mettere qualcosa della complexity?

Some constrained variants of the minimum path cover problem
have been introduced in the past by Ntafos and Hakimi in the context of
software testing~\cite{Ntafos1979} and appear to be relevant for some
sequence reconstruction problems of recent interest in bioinformatics.
More precisely, in software testing each procedure to be tested is
modeled by a graph where vertices correspond to single instructions
and two vertices are connected if the corresponding instructions are
executed sequentially.
The test of the procedure should check each instruction at least once,
hence a minimum path cover of the graph represents a minimum set of
execution flows that allows to test all the instructions.
Clearly, not all the execution flows are possible.
For this reason, Ntafos and Hakimi proposed the concept of required
pairs, which are pairs of vertices that a feasible solution must include
in a path, and that of impossible pairs, which are pairs of
vertices that a feasible solution must not include in the same path.
In particular, one of the problems introduced by Ntafos and Hakimi is
the \emph{Minimum Required Pairs Cover} (MinRPC) problem
where, given a DAG and a set of required pairs, the goal is to
compute a minimum set of paths \emph{covering} all the
required pairs, \ie a minimum set of paths such that, for each
required pair, at least one path contains both vertices of the pair.

The concept of required pairs is also relevant for
% the mentioned applications of
sequence reconstruction problems in bioinformatics, as short
reads are often sequenced in pairs (\emph{paired-end reads}) and these
pairs of reads must align to a single genetic sequence.
As a consequence, each pair of vertices corresponding to paired-end
reads must belong to the same path of the cover.
Paired-end reads provide valuable information that, in principle,
could greatly improve the accuracy of the reconstruction.
However, they are often used only to filter out the reconstructed sequences
that do not meet such constraints, instead of directly exploiting them
during the reconstruction process.
Notice that MinRPC asks for a solution that covers only the required
pairs, while in bioinformatics we are also interested in covering all
the vertices.
For this reason, we consider a variant of the Minimum
Path Cover problem, called \emph{Minimum Path Cover with Required Pairs}
(MinPCRP), that, given a DAG and a set of required pairs, asks for a
minimum set of paths covering all the vertices and all the required
pairs.
Clearly, MinPCRP is closely related to MinRPC.
In fact, as we show in Section~\ref{sec:pre}, the same reduction used
in~\cite{Ntafos1979} to prove the NP-hardness of MinRPC can be applied
to our problem, leading to its intractability.
% used for the same goal on our problem.

% In this paper, we continue the analysis given in~\cite{Ntafos1979} by studying how the
% computational complexity of the MinPCRP problem is influenced by the
% size of the solution, \ie the minimum number of paths covering all
% the vertices and all the required pairs.
In this paper, we continue the analysis of~\cite{Ntafos1979} by
studying the complexity of path covering problems with required
pairs. More precisely, we study how the complexity
of these problems is influenced by two parameters relevant for the sequence
reconstruction applications in bioinformatics: (1)~the minimum number of
paths covering all the vertices and all the required pairs and (2)~the
maximum \emph{overlapping degree} (defined later).
In the bioinformatics applications we discussed, the first
parameter---the number of covering paths---is often small, thus an
algorithm exponential in the size of the solution could be of interest.
The second parameter we consider in this paper, the maximum
overlapping degree, can be informally defined as follows. 
Two required pairs overlap when there exists a path that
connects the vertices of the pairs, and the path cannot be split in two
disjoint subpaths that separately connect the vertices of the two pairs.
Then, the overlapping degree of a required pair is the number of
required pairs that overlap with it.
In the sequence reconstruction applications, as the distance between two
paired-end reads is fixed, the maximum overlapping degree is small compared
to the number of vertices, hence it is a natural parameter for investigating the
computational complexity of the problem.

First, we investigate how the computational complexity of MinPCRP is
influenced by the first parameter.
In this paper we prove that it is NP-complete to decide if there
exists a solution of MinPCRP consisting of at most three paths
(via a reduction from the $3$-Coloring problem).
We complement this result by giving
a polynomial-time algorithm for
computing a solution with at most $2$ paths, thus establishing a sharp
tractability borderline for MinPCRP when parameterized by the size of
the solution.
These results significantly improve the hardness result that Ntafos and
Hakimi~\cite{Ntafos1979} presented for MinRPC (and that holds also for
MinPCRP), where the solution contains a number of paths which is
polynomial in the size of the input.

Then, we investigate how the computational complexity of MinPCRP is
influenced by the second parameter, the overlapping degree.
Unfortunately, MinPCRP is NP-hard even if the maximum overlapping degree
is 0.
In fact, this can be easily obtained by modifying the reduction
presented in~\cite{Ntafos1979} to hold also for restricted instances of
MinPCRP with no overlapping required pairs.
A natural heuristic approach for solving MinPCRP is the one which computes a
solution by iteratively adding a path that covers a maximum set of
required pairs not yet covered by a path of the solution.
This approach leads to a natural combinatorial problem, the
\emph{Maximum Required Pairs with Single Path} (MaxRPSP) problem, that,
given a DAG and a set of required pairs, asks for a path that covers the
maximum number of required pairs.
We investigate the complexity of MaxRPSP and we show that it is not
only NP-hard, but also W[1]-hard when the parameter is the number of
covered required pairs.
Similarly as MinPCRP, we consider the MaxRPSP problem parameterized by
the maximum overlapping degree but, differently from MinPCRP, we give a
fixed-parameter algorithm for this case.
This positive result shows a gap between the complexity of MaxRPSP and
the complexity of MinPCRP when parameterized by the maximum overlapping
degree. %%, since MaxRPSP is fixed-parameter tractable, while MinPCRP is
%%NP-hard even if the maximum overlapping degree is 0.

The rest of the paper is organized as follows.
First, in Section~\ref{sec:pre} we give some preliminary notions and we
introduce the formal definitions of the two problems.
In Section~\ref{sec:sharpTract}, we investigate the computational
complexity of MinPCRP when the solution consists of a constant number of
paths: we show that it is NP-complete to decide if there exists a solution of
MinPCRP consisting of at most three paths, while the existence of a
solution consisting of at most two paths can be computed in polynomial
time.
% Notice that the former result rules out also the possibility to have
% fixed-parameter algorithms for MinPCRP when the parameter is the size of
% the solution.
In Section~\ref{sec:maxrpsp}, we investigate the computational
complexity of MaxRPSP: we prove its W[1]-hardness when the parameter is
the number of required pairs covered by the path
(Section~\ref{sec:maxrpsp:W-hard}) and we give
a fixed-parameter algorithm when the parameter is the maximum
overlapping degree (Section~\ref{sec:maxrpsp:FPTalgo}).

% Due to the page limit, some of the proofs are omitted.

\section{Preliminaries}
\label{sec:pre}

In this section, we introduce the basic notions used in the rest of
the paper and we formally define the two combinatorial problems we are interested in.

While our problems deal with directed graphs,
we consider both directed and undirected graphs.
We denote an \emph{undirected graph} as $G=(V,E)$ where $V$ is the set
of vertices and $E$ is the set of (undirected)
edges, and a \emph{directed graph} as $D=(N,A)$ where $N$ is the set of
vertices and $A$ is the set of (directed) arcs.
We denote an edge of $G=(V,E)$ as $\{v,u\}\in E$ where $v,u\in
V$. Moreover, we denote an arc of $D=(N,A)$ as $(v,u)\in A$ where
$v,u\in N$.

Given a directed graph $D=(N,A)$, a \emph{path} $\pi$ from vertex $v$ to
vertex $u$, denoted as $vu$-path, is a sequence of vertices
$\langle v_1,\dots,v_n\rangle$ such that $(v_i,v_{i+1})\in A$,
$v=v_1$ and $u=v_n$.
We say that a vertex $v$
\emph{belongs to} a path $\pi=\langle v_1,\dots,v_n\rangle$, denoted
as $v\in \pi$,
if $v = v_i$, for some $1 \leq i \leq n$.
Given a path $\pi=\langle v_1,\dots,v_n\rangle$,
we say that a path $\pi'= \langle v_i, v_{i+1},\dots,v_{j-1},v_j\rangle$,
with $1 \leq i \leq j \leq n$, is a subpath of $\pi$.
Given a set $N'\subseteq N$ of
vertices, a path $\pi$ \emph{covers} $N'$ if every vertex of $N'$
belongs to $\pi$.

In the paper, we consider a set $R$ of pairs of vertices in $N$. We denote each pair
as $\pair{v_i}{v_j}$, to avoid ambiguity with the notations of edges and arcs.

Now, we are able to give the definitions of the combinatorial problems we
are interested in.

\begin{problem}{\emph{Minimum Path Cover with Required Pairs}
    (MinPCRP)}\\*
\noindent
\textit{Input:}
a directed acyclic graph $D=(N,A)$, a source $s\in N$, a sink
$t\in N$, and a set $R=\{ \pair{v_x}{v_y} \mid v_x,v_y\in N, v_x\neq
v_y\}$ of required pairs.\\*
\textit{Output:} a minimum cardinality set
$\Pi=\{\pi_1,\dots,\pi_n\}$ of directed $st$-paths such that every
vertex $v\in N$ belongs to at least one $st$-path $\pi_i \in \Pi$ and
every required pair $\pair{v_x}{v_y} \in R$ belongs to at least one $st$-path
$\pi_i\in \Pi$, i.e.~$v_x$, $v_y$ belongs to $\pi_i$.
\end{problem}

\begin{problem}{\emph{Maximum Required Pairs with Single Path}
  (MaxRPSP)}\\*
\noindent
\textit{Input:} a directed acyclic graph $D=(N,A)$, a source $s\in N$, a sink
$t\in N$ and a set $R=\{ \pair{v_x}{v_y} \mid v_x,v_y\in N, v_x\neq
v_y\}$ of required pairs.\\*
\textit{Output:} an $st$-path $\pi$ that
covers  a set $R'=\{ \pair{v_x}{v_y} \mid v_x,v_y\in \pi
\}\subseteq R$ of maximum cardinality.
\end{problem}

Two required pairs $\pair{u'}{v'}$ and $\pair{u''}{v''}$ in $R$
\emph{overlap} if there exists a path $\pi$ in $D$ such that the four
vertices appear in $\pi$ in one of the following orders (assuming that the
vertex $u'$ appears before $u''$ in $\pi$), where $v'$ and $u''$ are two distinct
vertices of $G$ (see Fig.~\ref{fig:ex:overl-pair}): 
\begin{compactitem}
\item $\langle u', u'', v', v'' \rangle$ (the two required pairs are \emph{alternated});
\item $\langle u', u'', v'', v' \rangle$ (the required pair $\pair{u''}{v''}$ is \emph{nested} in $\pair{u'}{v'}$ ).
\end{compactitem}
Notice that, from this definition, the required pairs $\pair{x}{y}$ and $\pair{y}{z}$
do not overlap.
% It follows from this definition of overlapping pairs that the required
% pairs $\pair{x}{y}$, $\pair{y}{z}$ cannot overlap.

Finally, consider a required pair $\pair{u'}{v'}$ of $R$.
We define the \emph{overlapping degree}
of $\pair{u'}{v'}$ as the number of required pairs in $R$ that overlap with
$\pair{u'}{v'}$.

%We say that $\pair{u'}{v'}$ is \emph{nested} in $\pair{u''}{v''}$ if the
%shortest subpath $\pi'$ of $\pi$ covering $\pair{u'}{v'}$ is a subpath
%of the shortest subpath $\pi''$ of $\pi$ covering $\pair{u''}{v''}$.
%Otherwise we say that they are \emph{alternated} .
% Notice that the nesting relation defines a partial order over the set of
% required pairs.
% Clearly, two maximal required pairs covered by a path $\pi$ of
% $D$ are either alternated or not overlapped.
%%\YP{Mettere nei preliminari?}
\begin{figure}[t]
  \centering
  \includegraphics[width=.99\linewidth]{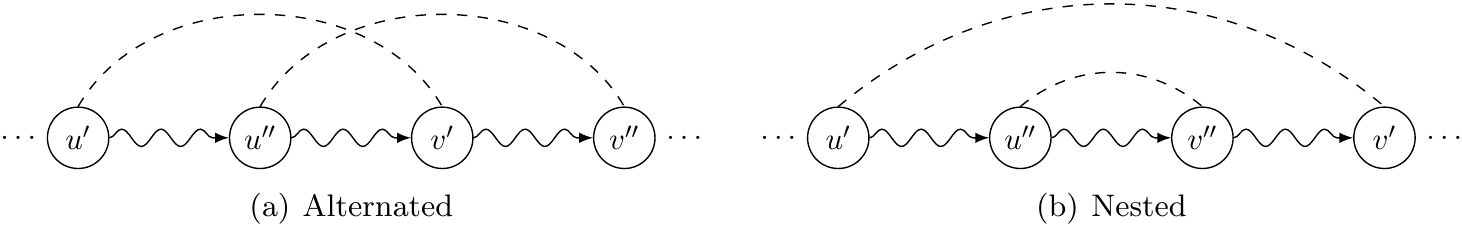}
  \caption{Examples of overlapping required pairs.
    The two required pairs $\pair{u'}{v'}$ and $\pair{u''}{v''}$ are
    represented by dashed lines. In (a) the required pairs are
    \emph{alternated}, while in (b) they are \emph{nested}.}
  \label{fig:ex:overl-pair}
\end{figure}

%%
%
%\begin{itemize}
%\item there exists a path $\pi$ that covers both  $\pair{u'}{v'}$ and
%  $\pair{u''}{v''}$
%
%\item $\pi$ contains a subpath that covers either the vertices $u'$,
%  $v'$, $u''$, and not $v''$, or the vertices $u'$, $v'$, $v''$, and not $u''$
%
%\end{itemize}

\paragraph{Hardness of MinPCRP.}
As we mentioned in the introduction,
MinPCRP is related to a combinatorial problem which has been studied
in the context of program testing~\cite{Ntafos1979}, where it is shown
to be NP-hard.
More precisely, given a directed acyclic graph $D=(N,A)$, a source $s\in
N$, a sink $t\in N$ and a set $R=\{ \pair{v_x}{v_y} \mid v_x,v_y\in N,
v_x\neq v_y\}$ of required pairs, the \emph{Minimum Required Pairs
  Cover} (MinRPC) problem asks for a minimum cardinality set
$\Pi=\{\pi_1,\dots,\pi_n\}$ of directed $st$-paths such that every
required pair $\pair{v_x}{v_y} \in R$ belongs to at least one $st$-path
$\pi_i\in \Pi$, i.e.~$v_x,v_y\in \pi_i$.

MinRPC can be easily reduced to MinPCRP due to the following property:
each vertex of the graph $D$ (input of MinRPC) must belong to at least
one required pair. Indeed, if this condition does not hold for some
vertex $v$, we can modify the graph $D$ by contracting $v$ (that is
removing $v$ and adding an edge $(u,z)$ to $A$, for each $u,z \in N$
such that $(u,v),(v,z) \in A$). This implies that, since in an
instance of MinRPC all the resulting vertices belong to some required
pair, a feasible solution of that problem must cover every
vertex of the graph. Then, a solution of MinRPC is also a solution of
MinPCRP, which implies that MinPCRP is NP-hard.

MinPCRP on \emph{directed} graphs (not necessarily acyclic) is as hard
as MinPCRP on DAGs.
In fact, since each strongly connected component can be covered with a
single path, we can replace them with single vertices, obtaining a DAG
and without changing the size of the solution.
Clearly, MinPCRP on general graphs and requiring that the covering
paths are simple is as hard as the Hamiltonian path problem, which is
NP-complete.

\section{A Sharp Tractability Borderline for MinPCRP}
\label{sec:sharpTract}

In this section, we investigate the computational complexity of MinPCRP
and we give a sharp tractability borderline for \kPCRP, the restriction of MinPCRP
where we ask whether there exist $k$ paths that cover all the vertices
of the graph and all the set of required pairs.
First, we show (Sect.~\ref{sec:PCRP3}) that \threePCRP is NP-complete.
This result implies that \kPCRP does not belong to the class XP
\footnote{We recall that the class XP contains those problems that,
  given a parameter $k$, can be solved in time $O(n^{f(k)})$}, so
it is probably hopeless to look for an algorithm having complexity $O(n^k)$, and hence
for a fixed-parameter algorithm in $k$.
We complement this result by giving (Sect.~\ref{sec:PCRP2}) a polynomial time algorithm for \twoPCRP,
thus defining a sharp borderline between tractable and intractable
instances of MinPCRP.

\subsection{Hardness of \threePCRP}
\label{sec:PCRP3}

In this section we show that \threePCRP is NP-complete.
We prove this result via a reduction from the well-known $3$-Coloring
(3C) problem which, given an undirected (connected) graph $G=(V,E)$,
asks for a coloring $c: V \rightarrow \{c_1, c_2, c_3\}$
of the vertices of $G$ with exactly $3$ colors, such that, for every
$\{v_i,v_j\} \in E$, we have $c(v_i) \neq c(v_j)$.

Starting from an undirected graph $G=(V,E)$ (instance of 3C),
we construct a corresponding instance $\langle D=(N,A), R\rangle$ of \threePCRP
as follows.
%First, notice that we assume that $G$ consists of a single connected component,
%otherwise the 3C problem can be solved independently on each connected component.
For every subset $\{v_i, v_j\}$ of cardinality $2$ of $V$, we define a
graph $D_{i,j}=(N_{i,j}, A_{i,j})$ (in the following we assume that, for each
$D_{i,j}$ associated with set $\{v_i, v_j\}$, $i < j$).
The vertex set $N_{i,j}$ is $\{ s^{i,j}, n_i^{i,j}, n_j^{i,j}, f^{i,j},
t^{i,j} \}$.
% defined as follows:
% \[
% N_{i,j} = \{ s^{i,j}, n_i^{i,j}, n_j^{i,j}, f^{i,j}, t^{i,j} \}
% \]
The set $A_{i,j}$ of arcs connecting the vertices of $N_{i,j}$ can have two possible
configurations, depending on the fact that $\{v_i, v_j\}$ belongs or
does not belong to $E$.
In the former case, that is $\{v_i, v_j\} \in E$,
% we say that
$D_{i,j}$ is in \emph{configuration (1)} (see Fig.~\ref{fig:MinPCRP3-config}\,(a))
and:
\[
A_{i,j}= \{(s^{i,j}, n_i^{i,j}), (s^{i,j} n_j^{i,j}), (s^{i,j}, f^{i,j}),
(n_i^{i,j}, t^{i,j}), (n_j^{i,j},t^{i,j}), (f^{i,j}, t^{i,j})\}
\]

In the latter case, that is $\{v_i, v_j\} \notin E$,
%we say that
$D_{i,j}$ is in \emph{configuration (2)} (see Fig.~\ref{fig:MinPCRP3-config}\,(b))
and:
\[
A_{i,j}= \{(s^{i,j}, n_i^{i,j}), (s^{i,j}, f^{i,j}), (n_i^{i,j}, n_j^{i,j}),  (n_j^{i,j},t^{i,j}),
(f^{i,j}, t^{i,j}) \}
\]

%%Figure~\ref{fig:MinPCRP3-config} reports the two configurations of the
%%subgraph $D_{i,j}=(N_{i,j}, A_{i,j})$ associated with the (unordered)
%%pair of vertices $\{v_i, v_j\} \in V \times V$.

\begin{figure}[t!]
\centering
\includegraphics[width=\linewidth]{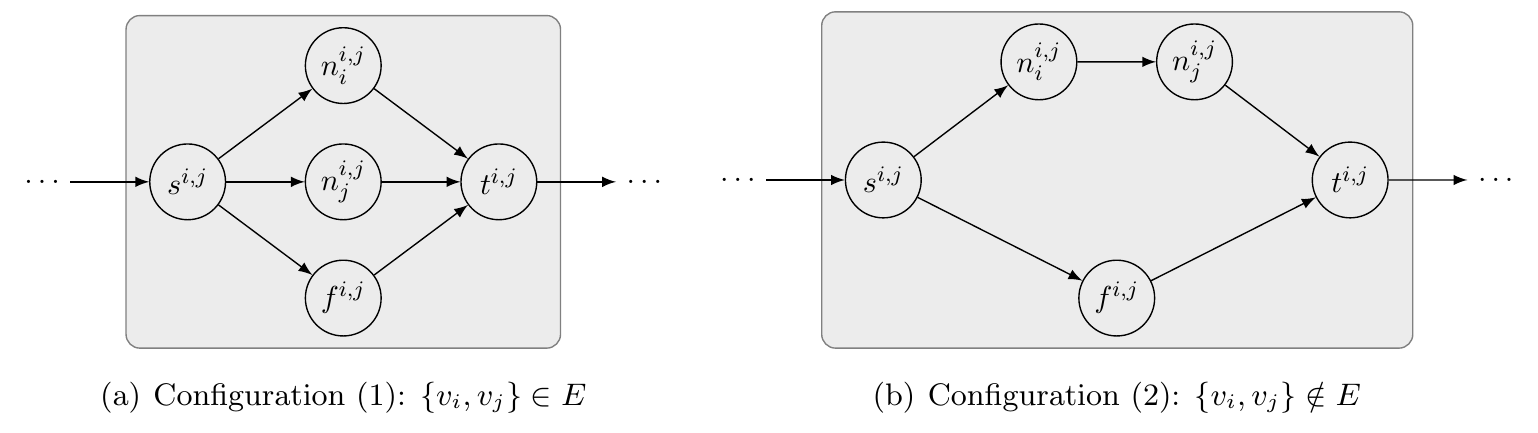}
\caption{Example of the two configurations of subgraph $D_{i,j}=(N_{i,j},
  A_{i,j})$ associated with a pair $\{v_i, v_j\}$ of vertices of a graph
  $G=(V,E)$.
}
\label{fig:MinPCRP3-config}
\end{figure}

The whole graph $D=(N,A)$ is constructed by concatenating the graphs
$D_{i,j}$ (for all $1 \le i < j \le n$) according to the lexicographic
order of their indices $i,j$.
The sink $t^{i,j}$ of each graph $D_{i,j}$ is connected to the source
$s^{i',j'}$ of the graph $D_{i',j'}$ which immediately follows
$D_{i',j'}$.
A distinguished vertex $s$ is connected to the source of $D_{1,2}$ (\ie
the first subgraph), while the sink of $D_{n-1,n}$ (\ie the last
subgraph) is connected to a second distinguished vertex $t$.
%
% The whole graph $D =(N,A)$ is constructed as follows.
% \[
% N = \{s,t\}  \cup (\bigcup_{i,j} N_{i,j})
% \]
% \[
% A = \{ (s,s^{1,2})\} \cup \{ (t^{n-1,n},t)\} \cup (\bigcup_{i,j} A_{i,j})
% \cup \{ (t^{i,j},s^{i,j+1}): 1 \leq i \leq n-1, i < j \leq n \}\]
% \[
% \cup \{ (t^{i,n},s^{i+1,i+2}): 1 \leq i \leq n-2\}
% \]
%
% Informally, the graph $D$ is constructed by connecting different subgraphs $D_{i,j}$,
% based on the lexicographic ordering of the indices, \ie connecting vertex $t^{i,j}$
% of a subgraph $D_{i,j}$, with vertex $s^{x,h}$ of the successive subgraph.
% Finally, the graph contains an arc between the source $s$ and $s^{1,2}$,
% and an arc between $t^{n-1,n}$ and the sink $t$.
Fig.~\ref{fig:MinPCRP3-graph} depicts such a construction.
% An example of graph $D=(N,A)$ is shown in .

The set $R$ of required pairs is defined as follows.
\[
R= \{ \pair{s}{f^{i,j}} \mid \{v_i, v_j\} \in E \} \cup \bigcup_{1 \le i
  \le n} R_i
\text{\quad where } R_i = \{
\pair{n_i^{i,j}}{n_i^{i,h}} \mid 1 \leq j \leq h \leq n \}
\]

\begin{figure}[t!]
\centering
\includegraphics[width=\linewidth]{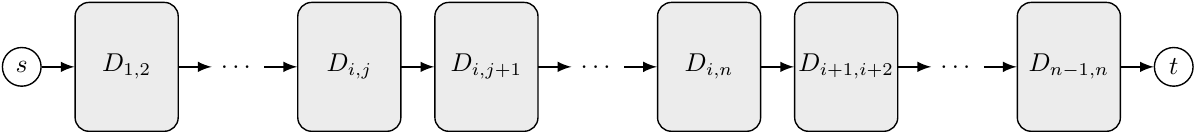}
\caption{Example of graph $D=(N, A)$ associated with graph
  $G=(V,E)$. Grey boxes represent subgraphs $D_{i,j}$ in one of the
  two possible configurations of Fig.~\ref{fig:MinPCRP3-config}.
}
\label{fig:MinPCRP3-graph}
\end{figure}

The following lemmas prove the correctness of the reduction.

\begin{lemma}
\label{lem:3-minpcrp-hard1}
Let $G=(V,E)$ be an undirected (connected) graph and let $\langle D=(N,A),
R\rangle$ be the corresponding instance of \threePCRP.
Then, given a $3$-coloring of $G$ we can compute in polynomial time three paths
of $D$ that cover all its vertices and every required pair in $R$.
\end{lemma}
\begin{proof}
Consider a $3$-coloring of $G$ and let $\{V_1, V_2, V_3\}$ be the
tri-partition of $V$ induced by the $3$-coloring.
% Without loss of generality, let $c_1$,
% $c_2$, $c_3$ be the three colors associated with the vertices of the graph $G$.
We show how to compute in polynomial time three paths $\pi_1$, $\pi_2$, $\pi_3$
that cover all the vertices of $D$ and every required pair in $R$.
% For each color $c_x$, $1 \leq x \leq 3$, let $V(c_x) \subseteq V$ be the set of vertices of $G$ associated with color $c_x$.
For each
$v_i \in V_c$, path $\pi_c$
passes through vertices $n_i^{i,j}$ of subgraphs $D_{i,j}$ for every
$v_j \in V$,
while for each subgraph $D_{p,q}$ such that $v_p, v_q \notin V_c$,
$\pi_c$ passes through verteces $f^{p,q}$.
Notice that each $\pi_c$ is well-defined, since when $n_i^{i,j}$, $n_j^{i,j}$ are associated
with the same color $c$, $D_{i,j}$ is in configuration (2), hence the path
can pass through both vertices $n_i^{i,j}$ and $n_j^{i,j}$.

We show that $\pi_1$, $\pi_2$, $\pi_3$ cover every required pair in $R$.
Notice that for each $\{v_i,v_j\} \in E$, since $v_i$ and $v_j$ have different
colors, by construction one of the paths $\pi_1$, $\pi_2$, $\pi_3$
passes through $n_i^{i,j}$, while another one % in $\pi(c_1)$, $\pi(c_2)$, $\pi(c_3)$
passes through $n_j^{i,j}$.
Now, we show that every required pair in $R_i$ is covered.
% By construction, each vertex $n_i^{i,j}$, $1 \leq i < j \leq n$ belongs to some path and,
% for some fixed $i$, each vertex $n_i^{i,j}$ associated with the same vertex $v_i$ belongs
% to the same path.
By construction, the vertices $n_i^{i,j}$ of $D$ associated with the
same vertex $v_i$ of $G$ belong to the same path $\pi_c$ where $c$ is
the color of $v_i$.
Therefore, all the required pairs in each $R_i$ are covered by one of
the three paths.
% Hence, the three defined paths $\pi(c_1)$, $\pi(c_2)$, $\pi(c_3)$ cover
% the required pairs in $R_i$.
Now, we show that $\pi_1$, $\pi_2$, $\pi_3$ cover the
required pairs $\{ \pair{s}{f^{i,j}} \mid 1 \leq i < j \leq n \}$.
Indeed, consider a subgraph $D_{i,j}$, and assume w.l.o.g.~that $v_i$
has color $c$ and that $v_j$ has color $c'$.
By construction, path $\pi_{c''}$ (with $c'' \notin \{c, c'\}$)
passes through $f^{i,j}$.
% Moreover, assume w.l.o.g.~that there exists a vertex $v_x \in V$ colored
% by $c''$.
% Otherwise we can assume that path $\pi_{c''}$ passes
% through all the vertices $f^{i,j}$.
% Now, since the color of $v_z$ is $c''$, by construction we have that
% $\pi_{c''}$ passes through $f^{i,j}$.
Then, $\pi_1$, $\pi_2$, $\pi_3$ cover all the required pairs
in $R$.

Finally, in order to show that all the vertices of $D$ are covered by
at least one path,
% $\pi(c_1)$, $\pi(c_2)$, $\pi(c_3)$,
the only vertices that might be not covered are $s^{i,j}$ and $t^{i,j}$,
for $1\leq i<j\leq n$, since they do not belong to any required pair.
However, these vertices are articulation points, hence all the three
paths necessarily pass through them.
\qed
\end{proof}
%%
%%
%% Proof moved to the appendix !!
%%
%%

\begin{lemma}
\label{lem:3-minpcrp-hard2}
Let  $G=(V,E)$ be an undirected graph and let $\langle D=(N,A),
R\rangle$ be the corresponding instance of \threePCRP.
Then, given three paths in $D$ that cover all its vertices and
every required pair in $R$ we can compute in polynomial time a
$3$-coloring of $G$.
\end{lemma}
\begin{proof}
Consider three paths $\pi_1$, $\pi_2$, $\pi_3$ of $D$ that cover  all
the vertices of $D$ and every required pair in $R$.
We define the corresponding $3$-coloring of the graph $G$, consisting
of the colors $c_1$, $c_2$, $c_3$.

First, we prove a property of the three paths $\pi_1$, $\pi_2$, $\pi_3$.
We show that, given a vertex $v_i \in V$, there exists at least one path
among $\pi_1$, $\pi_2$, $\pi_3$ that covers all the required pairs in $R_i$.
Consider a vertex $v_i \in V$.
Since $G$ is connected, it follows that there exists at least one vertex adjacent
to $v_i$, w.l.o.g.~$v_j$, such that $\{v_i, v_j\} \in E$. Now, consider the subgraph
$D_{i,j}$. By construction, since $D_{i,j}$ has a configuration (1),
a solution of MinPCRP must contain three different paths, each one passing through one
of the vertices $n_i^{i,j}$, $n_j^{i,j}$, $f^{i,j}$.
Now, assume that path $\pi_1$ passes through $n_i^{i,j}$.
Notice that $\pi_2$, $\pi_3$ cannot pass through $n_i^{i,j}$.
But then, since $\pi_1$ is the only path that covers $n_i^{i,j}$
and since $R_i$ contains a pair $\pair{n_i^{i,j}}{n_i^{i,h}}$,
for each $h \neq j$, it follows that
all the vertices $n_i^{i,h}$, $1 \leq h \leq n$, must belong to $\pi_1$.
%otherwise a pair $(n_i^{i,j}, n_i^{i,h}$, for some $h$ would not be covered.
It follows that, given a vertex $v_i \in V$, there exists one path
in $\{\pi_1, \pi_2, \pi_3  \}$ that covers all the required pairs in $R_i$.
Moreover, since all the three paths pass through the vertices
$s^{i,j}$ and $t^{i,j}$ for $1\leq i<j\leq n$, then all the vertices
of $D$ are covered by $\{\pi_1, \pi_2, \pi_3  \}$.

Now, we define a $3$-coloring of $G$, where $C= \{c_1, c_2, c_3 \}$ is
the set of colors.
If a required pair in $R_i$ is covered by a path $\pi_x$, $1\leq x
\leq 3$, then we assign the color $c_x$ to vertex $v_i$.
Notice that the coloring is feasible, that is $c(v_i) \neq c(v_j)$
when $\{ v_i,v_j \} \in E$. Indeed, consider two vertices $v_i$, $v_j$
associated with the same color, and consider the two corresponding sets $R_i$, $R_j$ of required pairs.
By construction, it follows that $R_i$, $R_j$ are covered by the same path,
say $\pi_1$.
Consider the subgraph $D_{i,j}$. Since $R_i$, $R_j$ are both covered by $\pi_1$,
it follows that $D_{i,j}$ must have a configuration (2), hence $\{v_i,v_j\} \notin E$.
Hence we have defined a $3$-coloring of $G$.
\qed
\end{proof}
%%
%%
%% Proof moved to the appendix !!
%%
%%

As a consequence of the previous lemmas, we can easily prove the following result.

\begin{theorem}
\label{teo:MinPCRP3hard}
\threePCRP is NP-complete.
\end{theorem}
\begin{proof}
The NP-hardness of \threePCRP follows directly from Lemma~\ref{lem:3-minpcrp-hard1}
and Lemma~\ref{lem:3-minpcrp-hard2} and from the NP-completeness of
3C~\cite{GareyJohnson}.
\threePCRP is in NP, since, given three paths $\pi_1$, $\pi_2$,
$\pi_3$, we can verify in polynomial time
that $\pi_1$, $\pi_2$, $\pi_3$ cover all the vertices of $D$ and that
every required pair in $R$ is covered by some path in $\{ \pi_1, \pi_2, \pi_3 \}$.
\qed
\end{proof}

\subsection{A Polynomial Time Algorithm for \twoPCRP}
\label{sec:PCRP2}

In this section we give a polynomial time algorithm for computing a solution of
\twoPCRP. %, the restriction of MinPCRP that asks if there exist
%$2$ paths that cover all the required pairs. This result nicely complements the NP-hardness
%of \threePCRP, given in Section~\ref{sec:PCRP3}.
Notice that $1$-PCRP can be easily solved in polynomial time, as
there exists a solution of $1$-PCRP if and only if
the reachability relation of the vertices of the input graph is a total order.

The algorithm for solving \twoPCRP is based on a
polynomial-time reduction to the $2$-Clique Partition problem, which, given an undirected graph
$G=(V,E)$, asks whether there exists a partition of $V$ in two sets $V_1$, $V_2$
both inducing a clique in $G$.
The $2$-Clique Partition problem is polynomial-time solvable~\cite[probl.~GT15]{GareyJohnson}.
To perform this reduction we assume that given $\langle D=(N,A), R
\rangle$, instance of \twoPCRP, every vertex of the graph $D$ belongs
to at least one required pair in $R$.
Otherwise, we add to $R$ the required pairs $\pair{s}{v_i}$ for all $v_i
\in N$ that do not belong to any required pair.
Therefore, a solution that covers all the required pairs
in $R$ covers also all the vertices, hence it is a feasible solution of
\twoPCRP.
Moreover, notice that this transformation does not
affect the solution of \twoPCRP, since all the paths start from $s$
and cover all the nodes of the graph, including the additional required
pairs.

The algorithm, starting from an instance $\langle D=(N,A), R \rangle$ of \twoPCRP,
computes in polynomial time a corresponding undirected graph $G=(V,E)$
where:
\begin{itemize}

\item $V = \{v_c\mid c \in R \}$

\item $E = \{ \{ v_{c_i}, v_{c_j} \} \mid \text{ there exists a path
  in $D$ that covers both $c_i$ and $c_j$} \}$

\end{itemize}

Given a set of required pairs $R' \subseteq R$, we denote by $V(R')$
the corresponding set of vertices of $G$ (\ie $V(R') = \{ v_c \mid c \in
R'\}$).

The algorithm is based on the following fundamental property.

\begin{lemma}\label{lem:2-case-clique}
Given an instance $\langle D=(N,A), R \rangle$ of \twoPCRP and the corresponding
graph $G=(V,E)$, then there exists a path $\pi$ that covers a set $R'$ of
required pairs if and only if $V(R')$ is a clique of $G$.
\end{lemma}
\begin{proof}
We prove the lemma by induction on the number $k$ of required pairs
(vertices, resp.) of $R'$ ($V(R')$, resp.).

When $k=0$ the lemma trivially holds, in fact having no required
pairs, \ie $R' = \varnothing$, induces an empty clique, \ie $V(\varnothing)$.

If $k=1$, then we can assume that there exists at least one path in
$D$ that covers the only required pair $c$ (otherwise no solution for
\twoPCRP exists), and $V(\{c\})$ induces a clique (of size $1$) in
$G$.

Now, assume that the lemma holds for every set of required pairs in $R$ (or set
of vertices of $G$) of size $k$,
we show that it holds also for a set of required pairs in $R$ (or set of vertices of $G$)
of size $k+1$.

Consider a path $\pi$ that covers a set $R'$ of $k+1$ required pairs.
We show that $V(R')$ induces a clique in $G$.
Let $c$ be a required pair in $R'$ and let $R''=R' \setminus \{ c \}$.
By induction hypothesis, $V(R'')$ is a clique of $G$.
Since $\pi$ passes through all the vertices belonging to required pairs of $R'$,
it follows
that there exists a path covering both the required pairs $c_i$ and
$c$, for every $c_i \in R''$. Hence, by construction,
$\{ v_{c_i} , v_c \} \in E$, for every $v_{c_i} \in V(R'')$, and so we
can conclude that $V(R')$ is a clique of $G$.

Consider a clique $V(R')$ of size $k+1$.
We show that there exists a path covering the set $R'$ of required
pairs.
Let $c=\pair{n_x}{n_y}$ be a required pair in $R'$ and let $R''=R'
\setminus \{ c \}$.
Clearly, $V(R'')$ induces a clique of size $k$ in $G$.
By induction hypothesis, there exists a path $\pi$ that covers all the
required pairs in $R''$.
Starting from path $\pi$, we can compute (in polynomial
time) a path $\pi'$ that covers $R'$.
Notice that either $n_x$ or $n_y$ does not belong to $\pi$,
otherwise $\pi$ would already cover the required pair $c$.
Assume w.l.o.g.~that $n_x$ does not belong to $\pi$.
Since for each vertex $v_{c_i} \in V(R'')$ there exists an edge $\{
v_{c_i}, v_c \}$ in $G$,
it follows that for each required pair $c_i \in R''$ there exists a
path $\pi_{c_i,c}$ connecting the
vertices of $c_i$ and $c$ (hence, in particular, $\pi_{c_i,c}$ passes through $n_x$).
Consider the vertices of $R''$ in the total order induced by $\pi$.
There exists a vertex $n_z$ in $\pi$ (possibly $n_z$ is the source $s$) such that
$n_z$ belongs to a required pair $c_z \in R''$, and $n_z$ is the last vertex in $\pi$ of
a required pair in $R''$ for which the path $\pi_{c_z,c}$ passes through $n_z$ and then $n_x$.
Let $n_{z+1}$ be the successor of $n_z$ in $\pi$ and let $c_{z+1}$ the
required pair containing $n_{z+1}$.
Notice that $\pi_{c,c_{z+1}}$ passes through $n_x$ and then
$n_{z+1}$.
Now, we can compute a path $\pi''$ by
concatenating the following subpaths: the subpath of $\pi$ from $s$ to
$n_z$, the subpath of $\pi_{c_z,c}$ from $n_z$ to
$n_x$, the subpath of $\pi_{c_{z+1},c}$ from $n_x$ to $n_{z+1}$ and the
subpath of $\pi$ from $n_{z+1}$ to $t$.
By construction it is easy to see that $\pi$ is a subpath of $\pi''$, hence
$\pi''$ covers all the required
pairs in $R''$ and passes through $n_x$.
A similar construction can be applied to $\pi''$ to cover $n_y$ (if
$n_y$ does not already belong to $\pi$), hence
obtaining a path $\pi'$ that covers every required pair in $R'$.
\qed
\end{proof}
%%
%%
%% Proof moved to the appendix !!
%%
%%

From Lemma~\ref{lem:2-case-clique}, it follows that, in order to
compute the existence of a solution of \twoPCRP over the instance
$\langle D=(N,A), R \rangle$ (in which every vertex of $D$ belongs to
at least one required pair in $R$),
we have to compute if there exists a $2$-Clique Partition of the
corresponding graph $G$.
%corresponding to $D$.
Computing the existence of a $2$-Clique Partition over a graph $G$ is
equivalent to compute if there exists a $2$-Coloring of the
complement graph $G'$ (hence deciding if $G'$ is bipartite), which is
well-known to be solvable in polynomial
time~\cite[probl.~GT15]{GareyJohnson}.
We can conclude that \twoPCRP can be decided in polynomial time.

\section{Parameterized Complexity of MaxRPSP}
\label{sec:maxrpsp}

In this section, we consider the parameterized complexity of
MaxRPSP. We show that although MaxRPSP is W[1]-hard when parameterized
by the optimum, \ie the number of required pairs covered by a single
path (Section~\ref{sec:maxrpsp:W-hard}), the problem becomes fixed-parameter
tractable if the maximum number of overlapping required pairs is a parameter
(Section~\ref{sec:maxrpsp:FPTalgo}).

\subsection{W[1]-hardness of MaxRPSP Parameterized by the Optimum}
\label{sec:maxrpsp:W-hard}

In this section, we investigate the parameterized complexity
of MaxRPSP when parameterized by the size of the solution, that is the maximum
number of required pairs covered by a single path, and we prove that the problem
is W[1]-hard (notice that this result implies
the NP-hardness of MaxRPSP). For details on parameterized reductions, we refer the
reader to~\cite{ParameterizedComplexity,Niedermeier}.

We prove this result via a parameterized reduction from the
Maximum Clique (MaxClique) problem.
Given an undirected graph $G=(V,E)$, MaxClique asks for a clique
$C\subseteq V$ of maximum size.
Here, we consider the two decision versions of MaxClique
and MaxRPSP, $h$-Clique and $k$-RPSP respectively,
parameterized by the sizes of the respective solutions.
For example, given an undirected graph $G=(V,E)$, the $h$-Clique problem consists of
deciding if there exists a clique $C\subseteq V$ of size $h$.
%On the other hand, given a directed acyclic graph
%$D=(N,A)$ and a set $R=\{ \pair{v_x}{v_y} \mid v_x,v_y\in N, v_x\neq
%v_y\}$ of required pairs, the $k$-RPSP problem consists of
%deciding if there exists a path $\pi$ that ``covers'' $k$ required
%pairs, i.e.~the set $R'=\{ \pair{v_x}{v_y} \mid v_x,v_y\in \pi \}\subseteq R$ has
%size $k$.
%%In order to prove that MaxRPSP is W[1]-hard, we show a reduction that
%preserves the parameter,
We recall that $h$-Clique is known to be W[1]-hard~\cite{tcs/DowneyF95}.

First, we start by showing how to construct an instance of $k$-RPSP starting
from an instance of $h$-Clique.
Given an (undirected) graph
$G=(V,E)$ with $n$ vertices $v_1, \ldots, v_n$, we construct the
associated directed acyclic graph $D=(N, A)$ as follows.
The set $N$ of vertices is defined as:
\[
N= \{ \vertex{i}{z} \mid v_i\in V, 1 \leq z \leq h \} \cup
  \{ s, t \}
\]
Informally, $N$ consists of two distinguished vertices $s,t$ and
of $h$ copies $\vertex{i}{1}, \ldots,\vertex{i}{h}$ of every vertex
$v_i$ of $G$.

The set of arcs $A$ is defined as:
\[
A= \{ (\vertex{i}{z}, \vertex{j}{z+1}) \mid \{v_i,v_j\} \in E,
  1 \leq z \leq h-1\} \cup
  \{(s,\vertex{i}{1}), (\vertex{i}{h},t) \mid v_i\in V\}
\]
Informally, we connect every two consecutive copies associated with vertices
that are adjacent in $G$, the source vertex $s$ to all the vertices
$\vertex{i}{1}$, with $1 \le i \le n$, and all the vertices $\vertex{i}{h}$, with
$1 \le i \le n$, to the sink vertex $t$.

The set $R$ of required pairs is defined as:
\[
R=\{ \pair{\vertex{i}{x}}{\vertex{j}{y}} \mid \{v_i,v_j\} \in E, 1\leq x<y\leq
h\}
\]

Informally, for each edge $\{v_i,v_j\}$ of $G$
there is a required pair $\pair{\vertex{i}{x}}{\vertex{j}{y}}$,
$1\leq x<y\leq h$, between every two different copies associated with $v_i$, $v_j$.

By construction, the vertices in $N$ (except for $s$ and $t$)
are partitioned into $h$ \emph{independent sets} $I_z=\{\vertex{i}{z} \mid
 1 \le i \le n \}$, with $1 \le z \le h$, each one containing a copy of
every vertex of $V$.
Moreover, the arcs of $A$ only connect two vertices of consecutive subsets
$I_z$ and $I_{z+1}$, with $1\leq z \le h-1$.
Figure~\ref{fig:graph-construction} presents an example of directed
graph $D$ associated with an undirected graph $G$.

\begin{figure}[t]
\centering
\includegraphics[width=\linewidth]{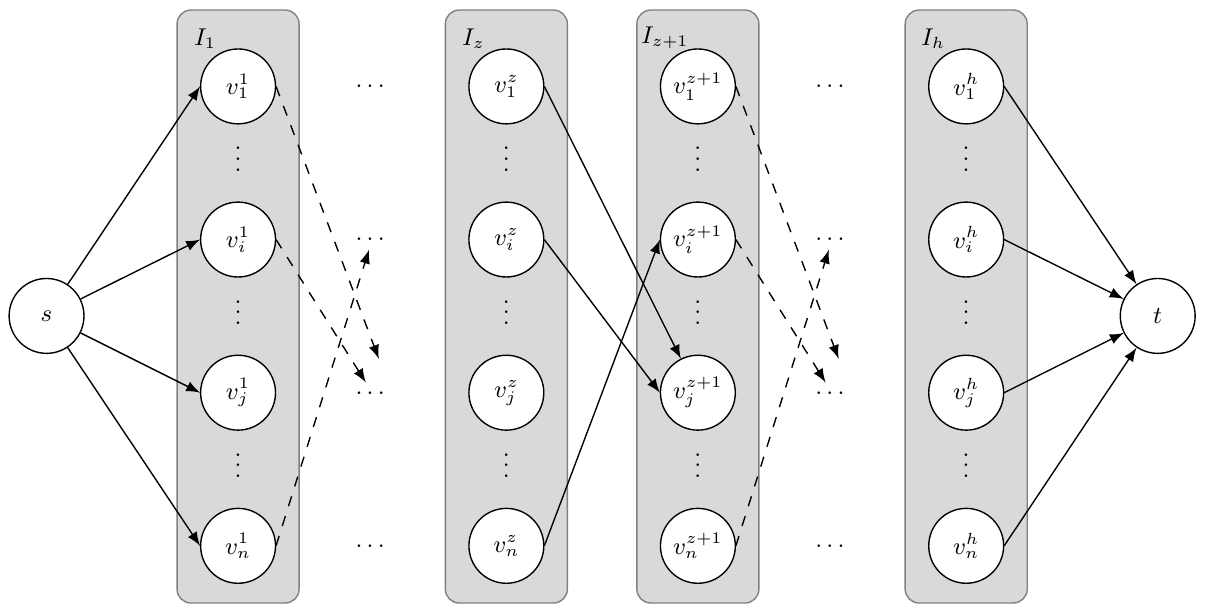}
\caption{Example of directed acyclic graph $D=(N,A)$ associated with
  an instance $G=(V,E)$ of the $h$-Clique problem.
  Each gray box highlight an independent set $I_z$ composed of one copy
  of the vertices in $V$.
  Edges $(\vertex{1}{z}, \vertex{j}{z+1})$, $(\vertex{i}{z},
  \vertex{j}{z+1})$, and $(\vertex{n}{z}, \vertex{i}{z+1})$ are some of
  the directed edges in $A$ associated with edges $\{v_1, v_j\}, \{v_i,
  v_j\}, \{v_i, v_n\} \in E$.}
\label{fig:graph-construction}
\end{figure}

Now, we are able to prove the main properties of the reduction.

\begin{lemma}\label{lem:max-clique}
Let $G=(V,E)$ be an undirected graph and $\langle
D=(N,A), R \rangle$ be the associated instance of $k$-RPSP.
Then:
(1) starting from an $h$-clique in $G$ we can compute in polynomial time
an $st$-path $\pi$ in $D$ that covers $\binom{h}{2}$
required pairs of $R$;
(2) starting from an $st$-path $\pi$ in $D$ that covers $\binom{h}{2}$
required pairs we can compute in polynomial time an $h$-clique in $G$.
\end{lemma}
\begin{proof}
(1) Starting from an $h$-clique $C$ in $G$ we show how to compute
a path $\pi$ in $D$ that covers $\binom{h}{2}$ required pairs of $R$.
Let $C=\{v_{i_1},\ldots,v_{i_h}\}$ be a clique of $G$ and let $\langle
v_{i_1},\dots,v_{i_h}\rangle$ be an arbitrary ordering of $C$.
Let $\pi_C=\langle s,\vertex{i_1}{1},\ldots,\vertex{i_h}{h},t\rangle$ be
a sequence of vertices obtained by selecting the vertex
$\vertex{i_z}{z}$ for each independent set $I_z$, with $1 \le z \le h$ (in
addition to vertices $s$ and $t$).
Since $C$ is a clique of $G$, by construction of $D$, every pair of
vertices $(\vertex{i_z}{z}, \vertex{i_{z+1}}{z+1})$ is connected by an
arc, hence $\pi_C$ is an $st$-path of $D$.
Moreover, the $st$-path $\pi_C$ covers exactly $\binom{h}{2}$ required
pairs of $R$ because, by construction of $R$, there exists a
pair between every two copies of vertices which are adjacent in $G$.
More precisely, since the clique $C$ has all the possible edges among
its $h$ vertices, the number of required pairs covered by the $st$-path
$\pi_C$ is $\binom{h}{2}$.

(2) Let $\pi$ be an $st$-path in $D$ that covers a set $R' \subseteq R$
of $\binom{h}{2}$ required pairs, then we show how to compute in
polynomial time an $h$-clique $C$ in $G$.
Notice that, by construction of $D$, the $st$-path $\pi$ must contain
exactly one vertex $\vertex{i}{z}$, $1 \leq i \leq n$ and $1 \le z \le
h$, for each independent set $I_z$ of $D$.
By construction of set $R$, each vertex $\vertex{i}{z}$
of $\pi$ appears in at most $h-1$ required pairs of $R'$.
Hence, the total number of required pairs covered by the $st$-path $\pi$, which
contains exactly $h$ inner vertices $\vertex{i}{z}$, is at most
$\frac{h(h-1)}{2}=\binom{h}{2}$.
Let $C$ be the set $\{v_i \mid \vertex{i}{z} \in \pi\setminus \{s,t\} \}$.
We claim that $C$ is an $h$-clique.
First, we prove that $C$ contains $h$ vertices.
Suppose to the contrary that $C$ has less than $h$ vertices.
Then, there exist two vertices $\vertex{i'}{x}$ and
$\vertex{i''}{y}$ of $\pi$ that correspond to the same vertex $v_i$ of $C$, that is
$i' = i'' = i$.
Since $\pair{\vertex{i}{x}}{\vertex{i}{y}} \notin R$, it follows that
each $\vertex{i}{x}$, $\vertex{i}{y}$ appears in at most $h-2$
required pairs of $R'$.
As a consequence, the total number of required pairs covered by the $st$-path
$\pi$ is strictly less than $\binom{h}{2}$, violating the initial
hypothesis that $\pi$ covers $\binom{h}{2}$ required pairs of $R$.
Hence $C$ contains $h$ vertices.
As all the internal vertices of $\pi$ (\ie all its vertices but $s$
and $t$) represent distinct vertices of $G$, then all the required pairs
covered by $\pi$ represent distinct edges of $G$.
The only undirected graph with $h$ vertices and $\binom{h}{2}$ edges is
the complete graph, hence $C$ is an $h$-clique of $G$.
% Moreover, by construction, it follows that the $\binom{h}{2}$ required pairs covered by
% $\pi$ correspond to $\binom{h}{2}$ distinct edges connecting the vertices in $C$.
% It follows that $C$ is an $h$-clique.
\qed
\end{proof}
%%
%%
%% Proof moved to the appendix !!
%%
%%

The W[1]-hardness of $k$-RPSP follows easily from Lemma~\ref{lem:max-clique}.

\begin{theorem}
$k$-RPSP is W[1]-hard when parameterized by the number of required pairs
covered by an $st$-path.
\end{theorem}
\begin{proof}
The result follows from Lemma~\ref{lem:max-clique} and from the
W[1]-hardness of $h$-Clique when parameterized by $h$~\cite{tcs/DowneyF95}.
\qed
\end{proof}

\subsection{An FPT Algorithm for MaxRPSP Parameterized by the Maximum
  Overlapping Degree}
\label{sec:maxrpsp:FPTalgo}

In this section we propose a parameterized algorithm for the MaxRPSP
problem, where the parameter is the maximum overlapping degree of the required pairs
in $R$.
For the rest of the section, let $\langle D=(N,A), R \rangle$ be an
instance of the MaxRPSP problem.
A required pair $\pair{u'}{v'} \in R$ is \emph{maximal} if it is not
nested in another required pair $\pair{u''}{v''}$.

% Since $D$ is a DAG, the reachability relation of its vertices
% is a partial order and we say that $v' \reach v''$ if $v''$ is reachable
% from $v'$.\YP{Controllare se serve realmente il simbolo $\reach$}
For ease of exposition, we fix an order of the required
pairs in $R$ and we represent the $i$-th required pair of the
ordering as $\pair{v^1_i}{v^2_i}$.
Whenever no confusion arises, we will refer to that required pair as
\emph{$i$-pair}.
Intuitively, we want that the order of the required pairs is
``compatible'' with the topological order of the vertices.
More formally, given two distinct required pairs $\pair{v^1_j}{v^2_j}$
and $\pair{v^1_i}{v^2_i}$ with $j < i$, then pair
$\pair{v^1_j}{v^2_j}$ is nested in $\pair{v^1_i}{v^2_i}$ or there does
not exist a path $\pi$ from $s$ to $v^2_j$ that covers both the required
pairs (that is, $\pi$ passes through $v^2_j$ before $v^2_i$).
% More formally, given two distinct required pairs $\pair{v^1_i}{v^2_i}$
% and $\pair{v^1_j}{v^2_j}$ with $i < j$, then we have that
% pair $\pair{v^1_i}{v^2_i}$ is nested in $\pair{v^1_j}{v^2_j}$
% or either $v^1_i$ is not reachable from $v^1_j$ or $v^2_i$ is not
% reachable from $v^2_j$.
Clearly, an order that satisfies this condition can be easily computed
from the topological order of the vertices.

% Given a required pair $\pair{v^1_i}{v^2_i}$, the set
% $OP(\pair{v^1_i}{v^2_i})$ is defined as the set of vertices $v$ such
% that $v$ belongs to a required pair overlapped with
% $\pair{v^1_i}{v^2_i}$ and such that $v^2_i$ is reachable from $v$.
% Notice that $OP(\pair{v^1_i}{v^2_i})$ always contains vertex $v^1_i$.
%
% \[
% OP(\pair{v^1_i}{v^2_i})=
% \{ v^1_i \} \cup \{v_j \mid v_j \reach v^2_i
%   \text{ and $j$-pair overlaps with $i$-pair} \}
% \]

% Property~\ref{prop:max-pairs} on maximal required
% pairs directly follows from the previous definition.

% \begin{property}\label{prop:max-pairs}
% Given two maximal required pairs $\pair{v^1_i}{v^2_i}, \pair{v^1_j}{v^2_j}
% \in R$ in a path $\pi$ of $D$ and assuming that vertex $v^1_i$ appears before
% $v^1_j$ in $\pi$, then they are either alternated ($\langle v^1_i,
% v^1_j, v^2_i, v^2_j \rangle$) or not overlapping ($\langle v^1_i, v^2_i,
% v^1_j, v^2_j \rangle$).
% \end{property}

We present a parameterized algorithm based on dynamic programming for
the MaxRPSP problem when the parameter $p$ is the maximum number of
overlapping required pairs.
In fact, we can decompose a path $\pi$, starting in $s$, ending in a
vertex $v$, and covering $k$ required pairs, into two subpaths: the first
one---$\pi_1$---starts in $s$, ends in a vertex $v'$, and covers
$k_1$ required pairs, while the other one---$\pi_2$---starts in $v'$, ends in
$v$, and covers the remaining $k_2=k-k_1$ required pairs (possibly
using vertices of $\pi_1$).
The key point to define the dynamic programming recurrence is that,
for each required pair $p$, we keep track the set of required pairs 
overlapping $p$ and covered by the path.
To this aim, for each required pair $\pair{v^1_i}{v^2_i}$, we define the
set $\OP(\pair{v^1_i}{v^2_i})$ as the set of vertices $v$ such that $v$
belongs to a required pair that overlaps $\pair{v^1_i}{v^2_i}$ and
such that $v^2_i$ is reachable from $v$.
By a slightly abuse of the notation, we consider that
$\OP(\pair{v^1_i}{v^2_i})$ always contains vertex $v^1_i$.

The recurrence relies on the following observation.
Let $\pi$ be a path covering a set $P$ of required pairs and let
$N(P)$ be the set of vertices belonging to the required pairs in $P$.
Consider two required pairs $\pair{v^1_i}{v^2_i}$ and
$\pair{v^1_j}{v^2_j}$ in $P$, with $j < i$.
Then, either
$\pair{v^1_j}{v^2_j}$ is nested in $\pair{v^1_i}{v^2_i}$ (hence
the fact that $\pi$ covers the pair $\pair{v^1_j}{v^2_j}$ can be checked
by the recurrence looking only at the required pairs that overlap with
$\pair{v^1_i}{v^2_i}$)
or pairs $\pair{v^1_i}{v^2_i}$ and $\pair{v^1_j}{v^2_j}$ are alternated.
In the latter case, since $\pair{v^1_i}{v^2_i}$ is in $P$, we only have
to consider the vertices in the set $N(P) \cap \OP(\pair{v^1_i}{v^2_i})
\cap \OP(\pair{v^1_j}{v^2_j})$.
Moreover, let $p_i$ be the number of required pairs that overlap the
required pair $\pair{v^1_i}{v^2_i}$, then $|\OP(\pair{v^1_i}{v^2_i})|$
is at most $2p_i$.
Hence, the cardinality of set $N(P) \cap \OP(\pair{v^1_i}{v^2_i})
\cap \OP(\pair{v^1_j}{v^2_j})$ is bounded by $2\max(p_i, p_j)$.
Moreover, given two sets $S$ and $S'$ of vertices such that $S \subseteq
\OP(\pair{v^1_i}{v^2_i})$ and $S' \subseteq \OP(\pair{v^1_j}{v^2_j})$, we
say that $S$ is in
\emph{agreement} with $S'$ if $S \cap (\OP(\pair{v^1_i}{v^2_i}) \cap
\OP(\pair{v^1_j}{v^2_j})) = S' \cap (\OP(\pair{v^1_i}{v^2_i}) \cap
\OP(\pair{v^1_j}{v^2_j}))$. Informally, when $S$ and $S'$ are in agreement,
they must contain the
same subset of vertices of $\OP(\pair{v^1_i}{v^2_i}) \cap
\OP(\pair{v^1_j}{v^2_j})$.

Let $P(\pair{v^1_i}{v^2_i}, S)$ denote the maximum number of required pairs
covered by a path $\pi$ ending in vertex $v^2_i$ and such that
the set $S \subseteq \OP(\pair{v^1_i}{v^2_i})$ is covered by $\pi$.
In the following we present the recurrence to compute $P(\pair{v^1_i}{v^2_i}, S)$.
For ease of exposition we only focus on vertices that appear as
second vertices of the required pairs.
In fact, paths that do not end in such vertices are not able to cover new
required pairs.
Furthermore, for simplicity, we consider the source $s$ as
the second vertex of a fictitious required pair (with index $0$)
$\pair{\bot}{s}$ which does not overlap any other required pair.
Such a fictitious required pair does not contribute to the total number
of required pairs covered by the path.

The recurrence is:
\begin{equation}
\label{eq:dyn-prog}
\footnotesize
P(\pair{v^1_i}{v^2_i}, S) =
\max_{\substack{\text{$\pair{v^1_j}{v^2_j}$\! not\! nested\! in\! $\pair{v^1_i}{v^2_i}$\! and\! $j<i$};\\
    S' \text{ in agreement with } S;\\
    \exists \text{ a path from } v^2_j \text{ to } v^2_i\\
    \text{ convering all vertices in } S \setminus S';
  }}
  \hspace{-2.1em}\left\{ P(\pair{v^1_j}{v^2_j}, S') +
    |Ov(\pair{v^1_i}{v^2_i}, S, S')|\right\}
\end{equation}
where
$Ov(\pair{v^1_i}{v^2_i}, S, S') =
\{\pair{v^1_h}{v^2_h} \mid \pair{v^1_h}{v^2_h} \text{ is nested in }
\pair{v^1_i}{v^2_i} \land v^1_h \in S \land v^2_h \in S\setminus S' \}$.
Notice that each required pair is assumed to be nested in itself.

%The base case of the recurrence is $P[v^2_i, \varnothing] = 0$, for each $v^2_i \in N$.
%Notice that we assume that $P[s,\varnothing] = 0$.
%|\{\pair{v^1_h}{v^2_h}: v^1_h, v^2_h \in S}|$.
The base case of the recurrence is $P(\pair{\bot}{s}, \varnothing) = 0$.
%|\{\pair{v^1_h}{v^2_h}: v^1_h, v^2_h \in S}|$.

The correctness of the recurrence derives from the following two lemmas.
\begin{lemma}
\label{lem:rpsp:dp1}
If $P(\pair{v^1_i}{v^2_i}, S) = k$, then there exists a path $\pi$ in $D$ ending
in $v^2_i$, such that every vertex in $S$ belongs to $\pi$ and the number
of required pairs covered by $\pi$ is $k$.
\end{lemma}
\begin{proof}
We prove the lemma by induction on the index $i$.
It is easy to see that the base case holds.
Assume that the lemma holds for index values less than $i$, we
prove that the lemma holds for $i$.
Let $P(\pair{v^1_i}{v^2_i}, S) = k$.
By Eq.~\eqref{eq:dyn-prog}, there exists a vertex $v^2_j$ with $j < i$,
such that $P(\pair{v^1_j}{v^2_j}, S') = k_1$ for some set $S'$
in agreement with $S$.
Assume that $|Ov(\pair{v^1_i}{v^2_i}, S, S')|=k_2$, with $k_1+k_2=k$.
By induction hypothesis, since $P(\pair{v^1_j}{v^2_j}, S')=k_1$, there exists a path
$\pi'$ ending in $v^2_j$, convering every vertex in $S'$, and
such that $\pi'$ covers $k_1$ required pairs.
Furthermore, the $k_2$ covered required pairs have at least one
vertex in $S \setminus S'$, hence the vertices of such required pairs belong
to a path $\pi''$ which starts in $v^2_j$ and ends in $v^2_i$ (path
$\pi''$ exists by hypothesis).
But then, the path obtained by the concatenation of $\pi'$ and
$\pi''$ covers $k_1+k_2$ required pairs.
\qed
\end{proof}

\begin{lemma}
\label{lem:rpsp:dp2}
Let $\pi$ be a path in $D$ ending in $v^2_i$ and covering $k$
required pairs.
Let $S$ be the set of all the vertices belonging to required
pairs covered by $\pi$ and overlapping $\pair{v^1_i}{v^2_i}$.
Then $P(\pair{v^1_i}{v^2_i}, S) \ge k$.
\end{lemma}
\begin{proof}
We prove the lemma by induction on the index $i$.
It is easy to see that the base case holds.
Assume that the lemma holds for index values less than $i$, we
prove that the lemma holds for $i$.
Let $\pi$ be a path, ending in $v^2_i$, that covers $k$ required pairs
and let $S$ be the set of vertices that belong to the required pairs covered
by $\pi$ and overlapping $\pair{v^1_i}{v^2_i}$.
We claim that $P(\pair{v^1_i}{v^2_i}, S) \ge k$.
Consider the rightmost vertex $v^2_j$ of $\pi$ such that $v^2_j$ belongs
to a required pair covered by $\pi$ and not nested in the $i$-pair.
Decompose path $\pi$ into two parts: one---$\pi'$---from $s$ to $v^2_j$, and
the other one---$\pi''$---from $v^2_j$ to $v^2_i$.
Let $S'$ be the set of vertices that belong to the required pairs covered
by $\pi$ and overlapping $\pair{v^1_j}{v^2_j}$.
Let $k'$ be the number of required pairs covered by $\pi'$ and $k''$ be
the number of the remaining required pairs covered by $\pi$ (that is,
$k=k'+k''$).
First, notice that $k'' = |Ov(\pair{v^1_i}{v^2_i}, S, S')|$.
By induction hypothesis $P(\pair{v^1_j}{v^2_j}, S')=k_1$ for some $k_1
\ge k'$.
Moreover, by construction, $S'$ is in agreement with $S$ and the subpath
of $\pi$ from $v^2_j$ to $v^2_i$ covers all the vertices in $S \setminus S'$.
As a consequence, by Eq.~\eqref{eq:dyn-prog}, $P(\pair{v^1_i}{v^2_i}, S)$ is at least
$k_1 + k'' \ge k'+k''=k$, which concludes the proof.
%
% Now, the set of required pairs covered by $\pi$ can be partitioned as follows: (1) the set
% $R'$ of required pairs consisting of vertices that both belong
% to a path $\pi'$ that starts in $s$ and
% ends in $v^2_j$ ($\pi'$ covers $k_1$ required pairs);
% (2) the set $R''$ of required pairs $\pair{v^1_h}{v^2_h}$ such that the vertex
% $v^2_h$ belongs to a path $\pi''$
% that starts in $v^2_j$ and ends in $v^2_i$. %%(the path $\pi$ covers
% %%$k_2$ required pairs, with $k=k_1+k_2$, having the rightmost vertex
% %%in $\pi''$).
%
% By induction hypothesis, the existence of path $\pi'$ implies
% $P[v^2_j, S']=k_1$, for some set $S'$ consisting of vertices
% that belong to required pairs in $\OP(\pair{v^1_j}{v^2_j}\})$ covered by $\pi'$.
% Moreover, consider the set
% $R''= \{\pair{v^1_h}{v^2_h}: \pair{v^1_h}{v^2_h} \textnormal{ is nested in } \pair{v^1_i}{v^2_i}\}$,
% of required pairs covered by $\pi$. Notice that for each pair $\pair{v^1_h}{v^2_h}$,
% the vertex $v^1_h$  either belongs to the path $\pi'$ (hence it must be in $S'$) or
% to $\pi''$, while vertex $v^2_h$ by construction belongs to $\pi''$.
% Since $|R''|=k_2$, with $k_1+k_2 =k$,
% applying Recurrence \ref{eq:dyn-prog} it holds $P[v^2_i,S]=k_1+k_2=k$.
\qed
\end{proof}

Let $p$ be the maximum number of overlapping required pairs in $D$ (that
is, $p=\max_i\{p_i\}$).
It follows
that the number of possible subsets $S$ is bounded by $O(2^p)$. Then,
each entry $P[v^2_i, S]$ requires time $O(2^p n)$ to be computed,
and, since there exist $O(2^p n)$ entries, the recurrence requires time $O(4^p n^2)$.
From Lemma~\ref{lem:rpsp:dp1} and Lemma~\ref{lem:rpsp:dp2}, it follows
that an optimal solution for MaxRPSP
can be obtained by looking for the maximum of the values $P[v^2_i, S]$.
Hence, the overall
time complexity of the algorithm is bounded by $O(4^p n^2)$.


\begin{thebibliography}{10}
\providecommand{\url}[1]{\texttt{#1}}
\providecommand{\urlprefix}{URL }

\bibitem{bao2013branch}
Bao, E., Jiang, T., Girke, T.: {BRANCH}: boosting {RNA-Seq} assemblies with
  partial or related genomic sequences. Bioinformatics  29(10),  1250--1259
  (2013)

\bibitem{Bonizzoni13Alg}
Bonizzoni, P., Dondi, R., Pirola, Y.: Maximum disjoint paths on edge-colored
  graphs: Approximability and tractability. Algorithms  6(1),  1--11 (2013)

\bibitem{Dilworth1950}
Dilworth, R.P.: A decomposition theorem for partially ordered sets. Annals of
  Mathematics  51(1),  161--166 (1950)

\bibitem{ParameterizedComplexity}
Downey, R., Fellows, M.: Parameterized Complexity. Springer Verlag (1999)

\bibitem{tcs/DowneyF95}
Downey, R.G., Fellows, M.R.: Fixed-parameter tractability and completeness {II}:
  On completeness for ${W[1]}$. Theoretical Computer Science  141(1{\&}2),
  109--131 (1995)

\bibitem{Beerenwinkel2008}
Eriksson, N., Pachter, L., Mitsuya, Y., Rhee, S.Y., Wang, C., Gharizadeh, B.,
  Ronaghi, M., Shafer, R.W., Beerenwinkel, N.: Viral population estimation
  using pyrosequencing. PLoS Comput Biol  4(5),  e1000074 (2008)

\bibitem{fordfulkerson}
Ford, L.R., Fulkerson, D.R.: Flows in Networks. Princeton university press
  (1962)

\bibitem{GareyJohnson}
Garey, M., Johnson, D.: Computer and Intractability: A Guide to the Theory of
  {NP}-completeness. W. H. Freeman (1979)

\bibitem{Niedermeier}
Niedermeier, R.: Invitation to Fixed-Parameter Algorithms. Oxford University
  Press (2006)

\bibitem{Ntafos1979}
Ntafos, S., Hakimi, S.: On path cover problems in digraphs and applications to
  program testing. IEEE Transactions on Software Engineering  5(5),
  520--529 (1979)

\bibitem{Trapnell2010}
Trapnell, C., Williams, B.A., Pertea, G., Mortazavi, A., Kwan, G., van Baren,
  M.J., Salzberg, S.L., Wold, B.J., Pachter, L.: {Transcript assembly and
  quantification by RNA-Seq reveals unannotated transcripts and isoform
  switching during cell differentiation}. Nature Biotechnology  28(5),
  516--520 (2010)

\bibitem{Wu12MaxCDP}
Wu, B.Y.: On the maximum disjoint paths problem on edge-colored graphs.
  Discrete Optimization  9(1),  50--57 (2012)

% \bibitem{Lee2}
% Xing, Y., Resch, A., Lee, C.: The multiassembly problem: reconstructing
%   multiple transcript isoforms from {EST} fragment mixtures. Genome Research
%   14(3),  426--441 (2004)

\end{thebibliography}
\end{document}